\documentclass{revtex4}
\usepackage{amsmath,amsthm}
\usepackage{graphicx}
\usepackage{epstopdf}
\usepackage{amssymb}
\usepackage{setspace}
\usepackage{amsmath}
\usepackage{amssymb}
\usepackage{amsfonts}
\newcommand{\bra}[1]{\langle#1|}
\newcommand{\ket}[1]{|#1\rangle}
\usepackage{braket}
\usepackage{amsfonts}
\newtheorem{theorem}{Theorem}[section]
\newtheorem{lemma}[theorem]{Lemma}
\newtheorem{proposition}[theorem]{Proposition}

\theoremstyle{remark}

\theoremstyle{definition}

\theoremstyle{example}

\theoremstyle{notation}

\begin{document}
\title{ The groupoid of bifractional transformations}
\author{S. Agyo, C. Lei, A. Vourdas\\
Department of Computer Science\\University of Bradford\\ Bradford BD7 1DP, UK}

\begin{abstract}
Bifractional transformations which lead to quantities that interpolate between other known quantities, are considered.
They do not form a group, and groupoids are used to described their mathematical structure.
Bifractional coherent states and bifractional Wigner functions are also defined.
The properties of the bifractional coherent states are studied.
The bifractional Wigner functions are used in generalizations of the Moyal star formalism.
A generalized Berezin formalism in this context, is also studied.

\end{abstract}

\maketitle
\section{Introduction}
Phase space methods \cite{Z1,Z2} are an important part of quantum mechanics.
Techniques like fractional Fourier transforms\cite{F1,F2,F3,F4}, coherent states\cite{C1,C2,C3}, analytic representations\cite{B1,B2,B3}, Wigner and Weyl functions, Moyal formalism\cite{M1,M2}, 
Berezin formalism\cite{BR1,BR2,BR3,BR4}, etc,
provide a deeper understanding of the nature of a quantum particle.

In this paper we introduce bifractional transforms, which lead to new quantities that interpolate between other known quantities.
 A preliminary version of this has been presented in \cite{ALV}.
Here we go deeper and expand these ideas as follows:
\begin{itemize}
\item
In section \ref{two} we review briefly for later use, the mathematical structure of groupoids  \cite{GR1,GR2,GR3, GR4}.
They are weaker but more general than groups. 
Applications of groupoids in physics include non-commutative geometry \cite{A1,A2}, quantum tomography \cite{A3}, etc. 
\item
 In section \ref{AA1} we introduce the bifractional displacement operators.
 They are two-dimensional  fractional Fourier transforms, but we stress that they are not a straightforward generalization of the one-dimensional 
 fractional Fourier transforms to the two-dimensional case (a technical point that we explain in section \ref{AA1}).
 The bifractional transformations do not form a group and we use groupoids to describe their mathematical structure.
 We also study the marginal properties of the bifractional transforms (section \ref{p22}).
\item
In section \ref{AA}, we act with the bifractional operators on the vacuum, and we get bifractional coherent states.
We study their analyticity properties and their resolution of the identity (proposition \ref{P1}).
We also study their overlaps and interpret the result in terms of a distance (proposition \ref{P2}).
The presentation emphasizes the difference between the formulas for standard coherent states, and the corresponding formulas for the bifractional coherent states.
\item
In section \ref{AA2} we study bifractional Wigner functions $A(\alpha, \beta;\theta_1,\theta_2|\rho)$.
Their marginal properties follow immediately from the marginal properties of the 
bifractional displacement operators in section \ref{p22}. In addition to that we give in proposition \ref{P100}, extra marginal properties that involve the 
$|A(\alpha, \beta;\theta_1,\theta_2|\rho)|^2$.
We also interpret physically the bifractional Wigner functions, as quantities which interpolate between quantum noise and quantum correlations.
\item
 In section \ref{AA3} we present briefly the Moyal star formalism for bifractional Wigner functions (proposition \ref{P15}).
\item
In section \ref{AA4} we present briefly the corresponding  Berezin formalism (proposition \ref{p13}).  
\end{itemize}

\section{Preliminaries: Groupoid ${\cal T}$ over ${\cal B}$}\label{two}
There are many cases where the concept of group is too strong for the description of a particular symmetry.
A weaker concept is the groupoid, which is designed for `variable symmetries'. 
Groups are special cases of groupoids.

A groupoid is a set ${\cal T}$ over a base set ${\cal B}$ such that
\begin{itemize}
\item
there two maps from ${\cal T}$ to ${\cal B}$
\begin{eqnarray}
s_1(t)=u_1;\;\;\;s_2(t)=u_2;\;\;\;\;t\in {\cal T};\;\;\;\;u_1,u_2\in {\cal B}
\end{eqnarray}
$u_1$ is the `source' of $t$, and $u_2$ the `target' of $t$.
$t$ can be viewed as an `arrow' which starts at $u_1$ and ends at $u_2$.
\item
A partial associative multiplication $t_1t_2$ is defined only in the case that $s_2(t_1)=s_1(t_2)$ 
(`the target of the first arrow is the same as the source of the second arrow').
\item
There is an involution (`inverse')
\begin{eqnarray}
t\;\rightarrow\;t^{-1};\;\;\;\;[t^{-1}]^{-1}=t.
\end{eqnarray}
\item
The elements $L_t=tt^{-1}$ and $R_t=t^{-1}t$ are called left and right identities, and they are in general different. 
Also
\begin{eqnarray}
L_t t=tR_t=t
\end{eqnarray}
Furthermore
\begin{eqnarray}
&&s_1(t^{-1}t)=s_2(t^{-1}t)=s_1(t^{-1})=s_2(t)\nonumber\\
&&s_1(tt^{-1})=s_2(tt^{-1})=s_2(t^{-1})=s_1(t).
\end{eqnarray}
The base set ${\cal B}$, is isomorphic to the set of all left identities and to the set of all right identities. 
\item
The set of all elements $t$ such that $s_1(t)=s_2(t)$ can be shown to form a group, called the isotropy group
\end{itemize}

In the special case that the base set ${\cal B}$ contains only one element 
($L_t=R_t={\bf 1}$), the multiplication is defined for all elements $t_1,t_2$, and the groupoid is a group. 
If for all $u_1,u_2\in {\cal B}$ there exists $t\in {\cal T}$ such that $s_1(t)=u_1$ and $s_2(t)=u_2$, the groupoid is called connected or transitive.

\section{Bifractional displacement operators}\label{AA1}
Let $\hat{x}, \hat{p}$ be the position and momentum operators of the harmonic oscillator. We consider the displacement operators 
\begin{eqnarray}
D(\alpha, \beta) = \exp(i\sqrt{2}\beta\hat{x}-i\sqrt{2}\alpha\hat{p}),
\end{eqnarray}
and the displaced parity operators
\begin{eqnarray}
\Pi(\alpha, \beta)&=&D\left (\frac{\alpha}{2},\frac{\beta}{2}\right )\Pi(0,0)D^\dagger \left (\frac{\alpha}{2},\frac{\beta}{2}\right )\nonumber\\
&=&D(\alpha, \beta)\Pi(0,0);\nonumber\\ \Pi(0,0)&=&\int dx\ket{x}\bra{-x}.
\end{eqnarray}
They are related through the two-dimensional Fourier transform (e.g., \cite{G1,G2,V1}) 
\begin{eqnarray}\label{7}
\Pi(\alpha, \beta)&=&\frac{1}{2\pi} \int  D(\alpha ', \beta ')\exp\left [i(\beta \alpha '-\beta '\alpha)\right ] d\alpha ' d\beta '\nonumber\\
&=& \int d\alpha 'd\beta '\Delta \left (\beta, \alpha ';\frac{\pi}{2}\right )
\Delta \left (\alpha,-\beta ';\frac{\pi}{2}\right ) D(\alpha ', \beta ')
\end{eqnarray}
We also consider the fractional Fourier transform:
\begin{eqnarray}\label{2}
\Delta (x,y;\theta)=\left [\frac {1+i\cot{\theta}}{2\pi}\right ]^{1/2}
\exp \left [\frac{-i(x^2+y^2)\cot\theta}{2}+\frac{ixy}{ \sin \theta}\right ]
\end{eqnarray}
Special cases of this are 
\begin{eqnarray}
&&\Delta (x,y;0)=\delta (x-y)\nonumber\\
&&\Delta \left (x,y;\frac{\pi}{2}\right )=\frac{\exp (ixy)}{(2\pi)^{1/2}}\nonumber\\
&&\Delta (x,y;\pi)=\delta (x+y).
\end{eqnarray}
We can prove that
\begin{eqnarray}\label{11}
\int dy\Delta (x,y;\theta _1)\Delta (y,z;\theta _2)=\Delta (x,z;\theta _1+\theta _2).
\end{eqnarray}
In \cite{ALV} we have generalized Eq.(\ref{7}) by replacing the Fourier transforms with fractional Fourier transforms. 
This led to the bifractional displacement operators
\begin{eqnarray}\label{1}
&&U(\alpha, \beta; \theta _1, \theta _2)
=|\cos(\theta_1-\theta_2)|^{\frac{1}{2}} \int d\alpha 'd\beta '\Delta \left (\beta,\alpha ';\theta _2\right )
\Delta \left (\alpha,-\beta ';\theta _1\right )
D(\alpha ', \beta ')
\end{eqnarray}
They are unitary operators.
The proof of unitarity is based on the integral
\begin{eqnarray}\label{24}
&&\int d\alpha 'd\beta 'd\alpha ''d\beta ''\Delta \left (\beta,\alpha ';\theta _2\right )
\Delta \left (\alpha,-\beta ';\theta _1\right )
D(\alpha ', \beta ')\Delta \left (-\beta,\alpha '';-\theta _2\right )
\Delta \left (-\alpha,-\beta '';-\theta _1\right )
D(\alpha '', \beta '')\nonumber\\&&=\frac{{\bf 1}}{|\cos (\theta_1-\theta _2)|}
\end{eqnarray}
In order to prove Eq.(\ref{24}) we use the relation 
\begin{eqnarray}
D(\alpha ', \beta ')D(\alpha '', \beta '')=D(\alpha '+\alpha '', \beta '+\beta '')\exp [i(\beta' \alpha ''-\alpha' \beta '')].
\end{eqnarray}
In the integration we are careful with the ordering of operators.
The variables $\alpha\rq{}, \beta \rq{}$ in Eq.(\ref{1}), are dual to each other and in this sense our fractional Fourier transform is not a 
straightforward generalization into two dimensions, of the one-dimensional fractional Fourier transform. 
Below we give a deeper explanation of the origin of the prefactor $|\cos(\theta_1-\theta_2)|^{\frac{1}{2}}$.

Since $\Delta (x,y;\theta +\pi)=\Delta (x,-y;\theta)$ it follows that 
\begin{eqnarray}
U(\alpha, \beta; \theta _1+\pi, \theta _2)=U(-\alpha, \beta; \theta _1, \theta _2);\;\;\;\;\;
U(\alpha, \beta; \theta _1, \theta _2+\pi)=U(\alpha, -\beta; \theta _1, \theta _2)
\end{eqnarray}
Therefore we can take $(\theta _1,\theta _2)\in {\mathfrak T}=[0,\pi)\times [0,\pi)-{\cal L}$, where ${\cal L}$ is the lines $\theta _1-\theta _2=\pm \frac{\pi}{2}$. We exclude them because in this case the prefactor in Eq.(\ref{1}) is zero.
The following are special cases:
\begin{eqnarray}
&&U\left (\alpha, \beta;0,0\right )=D(\beta, -\alpha)\nonumber\\
&&U\left (\alpha, \beta; \frac{\pi}{2}, \frac{\pi}{2}\right )=\Pi(\alpha, \beta)\nonumber\\
&&U\left (\alpha, \beta; \pi, \pi\right )=D(-\beta, \alpha).
\end{eqnarray}

\subsection{ $(\theta_1,\theta_2)$ axes in phase space and the origin of the factor $\cos(\theta_1-\theta_2)$}\label{159}

In most of the formulas throughout the paper we get the factor $\cos (\theta _1-\theta _2)$. 
The following arguments show the Jacobian nature of this factor, and also give a distance used later, in terms of coordinates in a non-orthogonal frame.

We consider an orthogonal frame $x-y$, and a non-orthogonal frame $x'-y'$ as shown in fig. \ref{f1}.
The `bifractional transform' in the present context is to rotate the $x$-axis by an angle $\theta _1$, and the $y$-axis by an angle $\theta _2$,
and change variables from $x,y$ to $x',y'$. Let $(x_0,y_0)$ and $(x_0',y_0')$ be the coordinates of a point in these two frames, correspondingly. 
With elementary trigonometry, we express the $(x_0',y_0')$ in terms of $(x_0,y_0)$ as follows
\begin{eqnarray}\label{app}
&&x_0'=Ax_0+By_0;\;\;\;\;y_0'=Cx_0+Dy_0\nonumber\\
&&A=\frac{\cos \theta _2}{\cos (\theta _1-\theta _2)};\;\;\;\;\;
B=\frac{\sin \theta _2}{\cos (\theta _1-\theta _2)}\nonumber\\
&&C=-\frac{\sin \theta _1}{\cos (\theta _1-\theta _2)};\;\;\;\;\;D=\frac{\cos \theta _1}{\cos (\theta _1-\theta _2)}.
\end{eqnarray}
Therefore the Jacobian corresponding to this change of variables is
\begin{eqnarray}
\frac{\partial (x_0',y_0')}{\partial (x_0,y_0)}=AD-BC=\frac{1}{\cos (\theta _1-\theta _2)}.
\end{eqnarray}

The distance of the point $(x',y')$ from the origin is given in terms of the coordinates in the non-orthogonal frame, by
\begin{eqnarray}\label{app1}
[d(x',y'|\theta _1,\theta _2)]^2=(x')^2+(y')^2+2x'y'\sin (\theta _1-\theta _2).
\end{eqnarray}

\subsection{The groupoid of transformations between the $U(\alpha,\beta;\theta _1, \theta _2)$}

The bifractional transformations do not form a group under multiplication.
It has been shown in \cite{ALV} that they are elements of  the semidirect product of the Heisenberg-Weyl group $HW$  by the
$SU(1,1)$ group of squeezing transformations: $HW\rtimes SU(1,1)$.
This has general elements of the type
\begin{alignat}{1}
S(a_1,a_2,a_3,a_4,a_5,a_6)=
\exp[a_1 \hat {x}^2+a_2\hat {p}^2+a_3(\hat {x}\hat {p}+\hat {p}\hat {x})+a_4\hat {x}+a_5\hat {p}+a_6{\bf 1}].
\end{alignat}
which depend on six parameters.
The operators $U(\alpha,\beta; \theta _{\alpha}, \theta _{\beta})$ depend only on four parameters, and they are
special cases of the operators $S(a_1,a_2,a_3,a_4,a_5,a_6)$. 

We describe the mathematical structure of $U(\alpha,\beta;\theta _1, \theta _2)$ with groupoids.
We consider the set of transformations 
\begin{eqnarray}\label{100}
{\cal B}=\{U(\alpha,\beta;\theta _1, \theta _2)\};\;\;\;\;(\alpha,\beta;\theta _1, \theta _2) \in {\mathbb R}\times {\mathbb R}\times {\mathfrak T} .
\end{eqnarray}
We also consider the map
\begin{eqnarray}\label{49}
{T}(\alpha, \beta ;\theta _1,\theta_2|\gamma, \delta ;\phi _1,\phi _2) :
U(\alpha,\beta;\theta _1, \theta _2)\;\; \rightarrow\;\;U(\gamma, \delta;\phi _1, \phi _2),
\end{eqnarray}
where
\begin{eqnarray}\label{3}
U(\gamma , \delta; \phi_1, \phi _2 )
=\frac{|\cos(\phi _1-\phi _2)|^{\frac{1}{2}}}{|\cos(\theta_1-\theta_2)|^{\frac{1}{2}}} \int d\alpha d\beta \Delta \left (\beta,\delta;\phi _2 -\theta _2\right )
\Delta \left (\alpha,\gamma;\phi _1-\theta _1\right )
U(\alpha , \beta ; \theta _1, \theta _2)
\end{eqnarray}
Eq.(\ref{3}) is a generalized version of Eq.(\ref{1}), which in the present notation is the map
\begin{eqnarray}
{T}{(\alpha ',\beta ';0,0|\alpha ,\beta ;\theta _1,\theta _2)} :U(\alpha ',\beta ';0,0)\;\; \rightarrow\;\;
U(\alpha, \beta;\theta_1,\theta_2),
\end{eqnarray}
The compatibility between the two, is shown in the first part of the proof
of proposition \ref{pro1} below.

We next consider the following notation for the composition
\begin{eqnarray}
&&[{T}{(\alpha,\beta;\theta _1,\theta_2|\gamma, \delta; \phi _1 \phi _2)} \circ {T}{(\gamma, \delta; \phi _1, \phi _2|\epsilon, \zeta; \psi _1, \psi _2)}][U(\alpha,\beta;\theta _1, \theta _2)]
\nonumber\\
&&={T}{(\gamma, \delta; \phi _1, \phi _2|\epsilon, \zeta; \psi _1, \psi _2)}[{T}{(\alpha,\beta;\theta _1,\theta_2|\gamma, \delta; \phi _1 \phi _2)} U(\alpha,\beta;\theta _1, \theta _2)]
\end{eqnarray}
The proposition below shows that the $T(\alpha,\beta;\theta _1,\theta_2|\gamma, \delta; \phi _1, \phi _2)$ form a groupoid:
\begin{proposition}\label{pro1}
The set $\{{T}{(\alpha,\beta;\theta _1,\theta_2|\gamma, \delta; \phi _1 ,\phi _2)}\}$ is a connected groupoid with base set ${\cal B}$ (in Eq.(\ref{100})), and with composition as multiplication.
The inverse of ${T}{(\alpha,\beta;\theta _1,\theta_2|\gamma, \delta; \phi _1 ,\phi _2)}$ is 
\begin{eqnarray}\label{400}
[{T}{(\alpha,\beta;\theta _1,\theta_2|\gamma, \delta; \phi _1, \phi _2)}]^{-1}={T}{(\gamma, \delta; \phi _1, \phi _2|\alpha,\beta;\theta _1,\theta_2)},
\end{eqnarray}
The left and right identities are ${T}{(\alpha,\beta;\theta _1,\theta_2|\alpha,\beta;\theta _1,\theta_2)}$ and
$T(\gamma, \delta; \phi _1 ,\phi _2|\gamma, \delta; \phi _1 ,\phi _2)$.

\end{proposition}
\begin{proof}
The proof consists of the following three parts:
\begin{itemize}
\item[(1)]
We prove that the following compatibility relation holds
\begin{eqnarray}\label{4}
{T}{(\alpha,\beta;\theta _1,\theta_2|\gamma, \delta; \phi _1 ,\phi _2)} \circ {T}{(\gamma, \delta; \phi _1, \phi _2|\epsilon, \zeta; \psi _1, \psi _2)}=
{T}{(\alpha,\beta;\theta _1,\theta_2|\epsilon, \zeta; \psi _1, \psi _2)}.
\end{eqnarray}
We start with the relations
\begin{eqnarray}\label{1a}
&&{T}{(\alpha,\beta;\theta _1,\theta_2|\gamma, \delta; \phi _1 ,\phi _2)}[U(\alpha,\beta;\theta_1,\theta_2)]=U(\gamma,\delta;\phi _1,\phi _2)\nonumber \\&&=
\frac{|\cos{(\phi _1-\phi _2)}|^\frac{1}{2}}{|\cos{(\theta_1-\theta_2)}|^\frac{1}{2}} \int d\alpha d\beta \Delta(\beta,\delta;\phi _2-\theta _2) \Delta(\alpha,\gamma;\phi _1-\theta_1) U(\alpha,\beta;\theta_1,\theta_2),
\end{eqnarray}
and
\begin{eqnarray}\label{2a}
&&T(\gamma,\delta; \phi _1,\phi _2|\epsilon,\zeta ; \psi _1, \psi _2) [U(\gamma,\delta;\phi _1,\phi _2)]= U(\epsilon,\zeta;\psi_1,\psi_2)\nonumber \\&&=
\frac{|\cos{(\psi _1-\psi _2)}|^\frac{1}{2}}{|\cos{(\phi _1-\phi _2)}|^\frac{1}{2}} \int d\gamma d\delta \Delta(\delta,\zeta;\psi _2-\phi _2) \Delta(\gamma,\epsilon;\psi _1-\phi _1) U(\gamma,\delta;\phi _1,\phi _2).
\end{eqnarray}
Inserting Eq.($\ref{1a}$) into  Eq.($\ref{2a}$) we get
\begin{eqnarray}\label{3a}
&&T(\gamma,\delta; \phi _1,\phi _2|\epsilon,\zeta ; \psi _1, \psi _2)[{T}{(\alpha,\beta;\theta _1,\theta_2|\gamma, \delta; \phi _1 ,\phi _2)} [U(\alpha,\beta;\theta _{1}, \theta _{2})]]
\nonumber \\&&=\frac{|\cos{(\psi_1-\psi_2)}|^\frac{1}{2}}{|\cos{(\theta_1-\theta_2)}|^\frac{1}{2}} \int d\gamma d\delta \Delta(\delta,\zeta;\psi_2-\phi_2) \Delta(\gamma,\epsilon;\psi_1-\phi_1)\nonumber \\&& \times
 d\alpha d\beta \Delta(\beta,\delta;\phi_2-\theta_2) \Delta(\alpha,\gamma;\phi_1-\theta_1)U(\alpha,\beta;\theta_1,\theta_2).
\end{eqnarray}
The compatibility relation of Eq.(\ref{4}) holds because using Eq.(\ref{11}) we show that Eq.(\ref{3a}) reduces to
\begin{eqnarray}
&&T(\gamma,\delta; \phi _1,\phi _2|\epsilon,\zeta ; \psi _1, \psi _2)[{T}{(\alpha,\beta;\theta _1,\theta_2|\gamma, \delta; \phi _1 ,\phi _2)} [U(\alpha,\beta;\theta_1, \theta_2)]]
 \nonumber \\ &&=\frac{|\cos{(\psi_1-\psi_2)}|^\frac{1}{2}}{|\cos{(\theta_1-\theta_2)}|^\frac{1}{2}}\int d\alpha d\beta \Delta(\beta,\zeta;\psi_2-\theta_2) \Delta(\alpha,\epsilon;\psi_1-\theta_1) U(\alpha,\beta;\theta_1,\theta_2)\nonumber\\&&=
T(\alpha,\beta; \theta_1,\theta_2|\epsilon,\zeta;\psi_1,\psi_2 ) [U(\alpha,\beta;\theta _{1}, \theta _{2})].
\end{eqnarray}

\item[(2)]

For the left and right identities, we first point out that in the special case that $\phi _1=\theta_1$ and $\phi _2=\theta_2$, the
$\Delta \left (\beta,\delta;0\right )$ and $\Delta \left (\alpha,\gamma;0\right )$
are delta functions, and therefore in this case $\alpha=\gamma$ and $\beta =\delta$ and ${T}{(\alpha,\beta;\theta _1,\theta_2 |\alpha, \beta; \theta _1,\theta_2)}$ is the identity map.
We also show that
\begin{eqnarray}
T(\alpha,\beta;\theta _1,\theta_2|\gamma,\delta;\phi_1,\phi_2)\circ {T}{(\gamma,\delta;\phi_1,\phi_2|\alpha,\beta;\theta _1,\theta_2)} ={T}{(\alpha,\beta;\theta _1,\theta_2|\alpha,\beta;\theta _1,\theta_2)} \nonumber\\ 
{T}{(\gamma,\delta;\phi_1,\phi_2|\alpha,\beta;\theta _1,\theta_2)}\circ T(\alpha,\beta;\theta _1,\theta_2|\gamma,\delta;\phi_1,\phi_2)={T}{(\gamma,\delta;\phi_1,\phi_2|\gamma,\delta;\phi_1,\phi_2)}.
\end{eqnarray}
The inverse, is an involution:
\begin{eqnarray}\label{40}
\{[{T}{(\alpha,\beta;\theta _1,\theta_2|\gamma,\delta;\phi_1,\phi_2)}]^{-1}\}^{-1}={T}{(\alpha,\beta;\theta _1,\theta_2|\gamma,\delta;\phi_1,\phi_2)}.
\end{eqnarray}
\item[(3)]
The above two parts show that ${\cal T}$ is a groupoid.
In fact it is a connected groupoid because any two elements $U(\gamma , \delta; \phi_1,\phi_2 )$, $U(\alpha , \beta ; \theta _1, \theta _2)$ 
in the set ${\cal U}$, are related through Eq.(\ref{3}).
\end{itemize}
\end{proof}

\subsection{Marginal properties for $U(\alpha, \beta; \theta _1, \theta _2)$}\label{p22}

The proposition below summarizes the marginal properties of $U(\alpha, \beta; \theta _1, \theta _2)$:
\begin{proposition}\label{pro10}
\mbox{}
\begin{itemize}
\item[(1)]
Integration of $U(\alpha, \beta; \theta _1, \theta _2)$ with respect to $\alpha$ gives
\begin{eqnarray}\label{qq10}
\int d\alpha U(\alpha, \beta; \theta _1, \theta _2)
=|\cos(\theta_1-\theta_2)|^{\frac{1}{2}} \int d\alpha 'd\beta '\Delta \left (\beta,\alpha ';\theta _2\right )
\Delta \left (0,\beta ';\frac{\pi}{2}-\theta _1\right ) D(\alpha ', \beta ')
\end{eqnarray}
\item[(2)]
Integration of $U(\alpha, \beta; \theta _1, \theta _2)$ with respect to $\beta$ gives
\begin{eqnarray}
\int d\beta U(\alpha, \beta; \theta _1, \theta _2)
=|\cos(\theta_1-\theta_2)|^{\frac{1}{2}} \int d\alpha 'd\beta '\Delta \left (0,\alpha ';\frac{\pi}{2}-\theta _2\right )
\Delta \left (\alpha,-\beta ';\theta _1\right )
D(\alpha ', \beta ')
\end{eqnarray}
\item[(3)]
Integration of $U(\alpha, \beta; \theta _1, \theta _2)$ with respect to both $\alpha$ and $\beta$ gives
\begin{eqnarray}\label{q23}
\int d\alpha d\beta U(\alpha, \beta; \theta _1, \theta _2)
= U(0,0; \frac{\pi}{2}-\theta _1, \frac{\pi}{2}- \theta _2)
\end{eqnarray}
\end{itemize}
\end{proposition}
\begin{proof}
\mbox{}
\begin{itemize}
\item[(1)]
\begin{eqnarray}\label{tyu}
\int U(\alpha,\beta;\theta_1, \theta_2) d\alpha &=& |\cos(\theta_1 - \theta_2)|^{\frac{1}{2}} \left[ \frac{1+i\cot\theta_1}{2\pi}\right]^{\frac{1}{2}} \int d\alpha 'd\beta '\Delta \left (\beta,\alpha ';\theta _2\right ) \exp\left[\frac{-i\beta'^2\cot\theta_1}{2}\right]
 D(\alpha ', \beta ') \nonumber \\
 &&\int d\alpha \exp\left[-\alpha^2 \left(\frac{i\cot\theta_1}{2} \right) -\alpha(\frac{i\beta'}{\sin\theta_1})\right] \nonumber \\
 &=&  |\cos(\theta_1 - \theta_2)|^{\frac{1}{2}} \left[1-i\tan\theta_1\right]^{\frac{1}{2}} \int d\alpha 'd\beta '\Delta \left (\beta,\alpha ';\theta_2\right ) \exp\left[\frac{i\beta'^2\tan\theta_1}{2}\right]
 D(\alpha ', \beta ') \nonumber \\
 &=&  |\cos(\theta_1 - \theta_2)|^{\frac{1}{2}} \int d\alpha 'd\beta' \Delta \left (\beta,\alpha ';\theta_2\right ) \Delta \left (0,\beta ';\frac{\pi}{2} - \theta _1\right )
 D(\alpha ', \beta ' ) 
\end{eqnarray}
We have used here the relation
\begin{eqnarray}\label{newK}
\Delta \left (0,\beta ';\frac{\pi}{2} - \theta \right ) = [1-i\tan\theta]^{\frac{1}{2}} \exp\left[ \frac{i\beta'^2\tan\theta}{2}\right]
\end{eqnarray}
\item[(2)]
This is proved in a similar way to the above.
\item[(3)]
We integrate Eq.(\ref{qq10}) with respect to $\beta$ and we prove Eq.(\ref{q23}).
\end{itemize}
\end{proof}

\section{Bifractional coherent states}\label{AA}
Acting on the $U(\alpha, \beta; \theta_1,\theta_2)$ on the vacuum $\ket{0}$ we get the 'bifractional coherent states': 
\begin{eqnarray}\label{cc}
\ket{\alpha, \beta; \theta_1,\theta_2} = U(\alpha, \beta; \theta_1,\theta_2)\ket{0}
\end{eqnarray}

We introduce another type of bifractional coherent states, which we 
call `R-bifractional coherent states\rq{}, and denote with the index R:
\begin{eqnarray}\label{dd}
\ket{\alpha, \beta; \theta_1,\theta_2}_R=U(0,0; \theta_1,\theta_2)\ket{\alpha, \beta}
\end{eqnarray}
 They are eigenstates of the annihilation operator
\begin{eqnarray}
b(\theta_1,\theta_2)=U(0,0; \theta_1,\theta_2)a[U(0,0; \theta_1,\theta_2)]^{\dagger}
\end{eqnarray}
and consequently they obey the resolution of the identity
\begin{eqnarray}\label{res2}
\frac{1}{2\pi}\int d\alpha d\beta \ket{\alpha, \beta; \theta_1,\theta_2}_{R\;R} \bra {\alpha, \beta; \theta_1,\theta_2}={\bf 1}.
\end{eqnarray}
The proposition below relates the bifractional coherent states in Eqs.(\ref{cc}), (\ref{dd}).
\begin{proposition}
\begin{eqnarray}\label{35}
&&\ket {\alpha, \beta; \theta_1,\theta_2 }_R= \ket {-\beta \cos \theta_1-\alpha \sin \theta_1, \alpha \cos \theta_2 - \beta \sin \theta_2; \theta_1, \theta_2}
\exp(iX)\nonumber\\
&&X=\frac{1}{4}(\beta ^2-\alpha ^2)[\sin (2\theta _2)-\sin (2 \theta _1)]+\alpha \beta (\cos ^2\theta _1-\cos ^2\theta _2)
\end{eqnarray}
\end{proposition}
\begin{proof}
We will prove that
\begin{eqnarray}\label{aa}
U(0,0; \theta_1,\theta_2) D(\alpha, \beta) &=& U(-\beta\cos\theta_1-\alpha\sin\theta_1, \alpha\cos\theta_2 - \beta\sin\theta_2; \theta_1,\theta_2)\exp(iX)
\end{eqnarray}
We write $U(0,0; \theta_1,\theta_2)D(\alpha, \beta)$ as
\begin{eqnarray}
U(0,0; \theta_1,\theta_2)D(\alpha, \beta) = |\cos(\theta_1-\theta_2)|^{1/2} &&\int d \alpha' d \beta' \exp\left[-\frac{i}{2}(\alpha'^2\cot\theta_2 + \beta'^2\cot\theta_1) \right] \nonumber \\
&& \times D(\alpha + \alpha', \beta+\beta') \exp[i\alpha\beta' -i\alpha' \beta] 
\end{eqnarray}

Changing variables and combining similar terms we show that,
\begin{eqnarray}
U(0,0; \theta_1,\theta_2)\ket{\alpha, \beta} = &&|\cos(\theta_1-\theta_2)|^{1/2}\exp(iX) \int d \gamma d \lambda \;\; \Delta(\gamma, \alpha\cos\theta_2-\beta\sin\theta_2; \theta_2)\nonumber\\& &\times \Delta(-\lambda, -\beta\cos\theta_1-\alpha\sin\theta_1; \theta_1) \ket{\gamma, \lambda}
\end{eqnarray}
This proves Eq($\ref{aa}$) and Eq.(\ref{35}).
\end{proof}

We will use the notation
\begin{eqnarray}\label{rrr}
&&\ket{\alpha, \beta; \theta _1, \theta _2}=\ket{w(\theta _1, \theta _2)};\;\;\;\;
w(\theta _1, \theta _2)= \frac{\alpha e_2 + i\beta e_1}{\cos(\theta_1-\theta_2)}\nonumber\\
&&e_2 = i\exp(-i\theta_2);\;\;\;\;\;e_1 = i\exp(-i\theta_1)
\end{eqnarray}
$w(\theta _1, \theta _2)$ can be written in terms of the $A,B,C,D$ in Eq.(\ref{app}) as
\begin{eqnarray}
w(\theta _1, \theta _2)=\alpha (B+iA)-\beta(D+iC)
\end{eqnarray}
$e_1$, $e_2$ are such that  the `analyticity part\rq{} of the following proposition, which presents the properties of the
bifractional coherent states, holds. 
\begin{proposition}\label{P1}
\mbox{}
\begin{itemize}
\item[(1)][analyticity]
The
$\exp \left (E(w|w^*)\right )\ket{w( \theta _1, \theta _2)}$, where
\begin{eqnarray} \label{p2}
E(w|w^*)= \frac{1}{2} |w|^2-\frac{1}{2}b(w^*)^2;\;\;\;\;\;
b= \frac{1}{4}[\exp(-i2\theta_1) -\exp(-i2\theta_2) ];\;\;\;\;E(w|w^*)=[ E(w^*|w)]^*
\end{eqnarray}
depends only on $w$, and does not depend on $w^*$.

\item[(2)][Resolution of the identity]

\begin{eqnarray}\label{vb}
\frac{1}{2\pi \cos(\theta_1-\theta_2)}\int d\alpha d\beta \ket{\alpha, \beta; \theta_1,\theta_2} \bra{\alpha, \beta; \theta_1,\theta_2} = 1
\end{eqnarray}
This can also be written as:
\begin{eqnarray}\label{104}
 \int \frac{d^2w}{2\pi}\ket{w(\theta _1, \theta _2)}\bra{w(\theta _1, \theta _2)}={\bf 1}
\end{eqnarray}
\end{itemize}
\end{proposition}
\begin{proof}
\mbox{}
\begin{itemize}
\item[(1)]
We write the 
$\exp \left (-\frac{1}{2}b(w^*)^2+\frac{1}{2}|w|^2\right )\ket{w(\theta _1, \theta _2)}$ as
\begin{eqnarray}
&&\exp \left (-\frac{1}{2}b(w^*)^2+\frac{1}{2}|w|^2]\right )\ket{w(\theta _1, \theta _2)}=
|\cos(\theta_1-\theta_2)|^{\frac{1}{2}} \int d^2\zeta A(w,\zeta)\exp(\zeta a^\dagger) \ket{0}\nonumber\\
&&A(w,\zeta)=\Delta \left (\mu _2,\zeta _R;\theta _2\right )
\Delta \left (\mu _1,-\zeta _I;\theta _1\right ) \exp\left [ -\frac{1}{2}b(w^*)^2+ \frac{1}{2} |w|^2-\frac{1}{2}|\zeta|^2\right ]\nonumber\\
&&\mu _1=\frac{1}{2}(e_1^* w + e_1 w^*) \nonumber \\
&&\mu _2= \frac{1}{2i}(e_2^* w - e_2 w^*)
\end{eqnarray}

Analyticity of coherent states gives
\begin{eqnarray}
|\cos(\theta_1-\theta_2)|^{\frac{1}{2}}  \int d^2\zeta A(w,\zeta)\frac{\partial }{\partial \zeta^*}[\exp(\zeta a^\dagger) \ket{0}]=0.
\end{eqnarray}
and integration by parts gives
\begin{eqnarray}\label{29}
|\cos(\theta_1-\theta_2)|^{\frac{1}{2}}  \int d^2\zeta \left [\frac{\partial }{\partial \zeta^*}A(w,\zeta)\right ]\exp(\zeta a^\dagger) \ket{0}=0.
\end{eqnarray}
We will prove that
\begin{eqnarray}\label{xc}
\frac{\partial }{\partial w^*}A(w,\zeta)=\frac{\partial}{\partial \zeta^*}A(w,\zeta).
\end{eqnarray}
Eq.(\ref{xc}) gives
\begin{eqnarray}
\frac{\partial}{\partial w^*} A(w,\zeta) &= & A(w,\zeta)  \left[-\frac{i \cot\theta_2}{4}(|e_2|^2 w - {e_2}^2 w^*) - \frac{e_2 \zeta_R }{2\sin\theta_2}\right] \nonumber \\
& + & A(w,\zeta) \left[ - \frac{i\cot\theta_1}{4}(|e_1|^2 w + e_1^2 w^*) -\frac{i e_1 \zeta_I}{2\sin\theta_1} \right] \nonumber \\
 &+ & A(w,\zeta) \left(-bw^* + \frac{1}{2}w\right ) 
\end{eqnarray}
and also,
\begin{eqnarray}
\frac{\partial}{\partial \zeta^*} A(w,\zeta) = &&A(w,\zeta)  \left[ -\frac{i e_2 \zeta_R}{2 \sin\theta_2} + \frac{1}{4\sin\theta_2}(e_2^* w - e_2 w^*)\right ]  \nonumber \\ 
  &&+A(w,\zeta)  \left[-\frac{i e_1 \zeta_I}{2 \sin\theta_1} + \frac{1}{4 \sin\theta_1}(e_1^* w + e_1 w^*)\right ]
\end{eqnarray}
Comparing coefficients we find the values of $e_1, e_2$ given in Eq(\ref{rrr}) and the value of $b$ given in Eq.(\ref{p2}).

We now insert Eq.(\ref{xc}) into Eq.(\ref{29}) and prove that
\begin{eqnarray}
\frac{\partial}{\partial w^*}\left [\exp \left (-\frac{1}{2}b(w^*)^2+ \frac{1}{2} |w|^2\right )\ket{w( \theta _1, \theta _2)}\right ]=0
\end{eqnarray}
\item[(2)]
Inserting Eq.(\ref{35}) into Eq(\ref{res2}) we prove the resolution of identity in Eq.(\ref{vb}).
Changing variables from $\alpha, \beta$ to $w(\theta _1, \theta _2)$ using Eq.(\ref{rrr}), we prove 
Eq.(\ref{104}).
\end{itemize}
\end{proof}

The following proposition, gives the overlap of two bifractional coherent states.
The square of the absolute value of this overlap is given in terms of a distance.
\begin{proposition}\label{P2}
\mbox{}
\begin{itemize}
\item[(1)]
The overlap of two of these coherent states is
\begin{eqnarray}\label{bas}
\braket{w(\theta _1, \theta _2) | v(\theta _1, \theta _2)} &= & \exp \left[-\frac{1}{2}|w(\theta _1, \theta _2)|^2-\frac{1}{2}|v(\theta _1, \theta _2)|^2 + w^*(\theta _1, \theta _2)v(\theta _1, \theta _2) \right] \nonumber \\
&\times & \exp \left[i\left[b^*w(\theta _1, \theta _2)^2 + b{w^*(\theta _1, \theta _2)}^2\right] \tan(\theta_1-\theta_2) \right] \nonumber \\
&\times & \exp \left[-i\left[b^*v(\theta _1, \theta _2)^2 + b{v^*(\theta _1, \theta _2)}^2\right]\tan(\theta_1-\theta_2) \right]
\end{eqnarray}
The last two factors are the `correction\rq{} to the usual result for the overlap of two coherent states.

\item[(2)]
\begin{eqnarray}\label{e2a}
|\langle {w(\theta _1, \theta _2)}\ket{v(\theta _1, \theta _2)}|^2=
\exp(-|w(\theta _1, \theta _2)-v(\theta _1, \theta _2)|^2)=
\exp\left\{ -\frac{[d(\alpha -\alpha ',\beta -\beta '|\theta_1,\theta_2)]^2 } {\cos^2(\theta_1-\theta_2)}\right \}
\end{eqnarray}
Here $d(\alpha -\alpha ',\beta -\beta '|\theta_1,\theta_2)$ is the distance discussed in the section \ref{159}.
The denominator ${\cos^2(\theta_1-\theta_2)}$ is a Jacobian as we change variables from an orthogonal to a non-orthogonal frame.
The $w(\theta _1, \theta _2)$  and $v(\theta _1, \theta _2)$ depend on $\alpha, \beta$ and $\alpha \rq{}, \beta \rq{}$ correspondingly, as in Eq.(\ref{rrr}).
\end{itemize}
\end{proposition}
\begin{proof}
\begin{itemize}
\item[(1)]
We prove Eq.(\ref{bas}) using Eq.(\ref{1}) for the two bi-fractional operators:
\begin{eqnarray}\label{kay}
&&\braket{\alpha,\beta;\theta_1,\theta_2|\alpha',\beta';\theta_1,\theta_2}=R\int   d\gamma d\lambda d\gamma' d\lambda'\nonumber\\
&&\times  \exp\left[-\frac{i}{2}(\gamma'^2-\gamma^2)\cot\theta_2 + i(\beta\gamma + \beta'\gamma')\csc\theta_2 -\frac{i}{2}(\lambda'^2-\lambda^2)\cot\theta_1 - i(\alpha\lambda + \alpha'\lambda')\csc\theta_1  \right] \nonumber \\ &&\times  \braket{0|D(\gamma,\lambda)D(\gamma',\lambda')|0}
\end{eqnarray}
where,
\begin{eqnarray}
R= [4\pi^2 \sin\theta_1 \sin\theta_2]^{-1} \exp{\left[\frac{i}{2}(\alpha^2\cot\theta_1 + \beta^2\cot\theta_2 -\alpha'^2\cot\theta_1 -\beta'^2\cot\theta_2)\right]}
\end{eqnarray}
Then we change notation using Eq.(\ref{rrr}). 

\item[(2)]
The first part  follows immediately from Eq.(\ref{bas}).
Then we use Eq.(\ref{app1}) in conjunction with Eq.(\ref{rrr}), and we prove the second part. 

\end{itemize}
\end{proof}

\section{Bifractional Wigner functions: interpolating between quantum noise and quantum correlations }\label{AA2}
For a density matrix $\rho$ we define the Wigner function $W(\alpha, \beta|\rho)$ and the Weyl function ${\widetilde W}(\alpha, \beta|\rho)$ as
\begin{eqnarray}
W(\alpha, \beta|\rho)={\rm Tr}[\rho \Pi(\alpha , \beta )];\;\;\;\;\;
{\widetilde W}(\alpha, \beta|\rho)={\rm Tr}[\rho D(\alpha , \beta )].
\end{eqnarray}
Using Eq.(\ref{7}) we show that the Wigner and Weyl functions are related through the two-dimensional Fourier transform:
\begin{eqnarray}\label{77}
&&W(\alpha, \beta|\rho)=\frac{1}{2\pi} \int  {\widetilde W}(\alpha ', \beta '|\rho)\exp\left [i(\beta \alpha '-\beta '\alpha)\right ] d\alpha ' d\beta '\nonumber\\
&&= \int d\alpha 'd\beta '\Delta \left (\beta, \alpha ';\frac{\pi}{2}\right )
\Delta \left (\alpha,-\beta ';\frac{\pi}{2}\right ) {\widetilde W}(\alpha ', \beta '|\rho).
\end{eqnarray}
In ref.\cite{ALV} we have generalized them into the bifractional Wigner function
\begin{eqnarray}\label{10}
A(\alpha, \beta;\theta_1,\theta_2|\rho) &=& {\rm Tr}(\rho  U(\alpha, \beta; \theta _1, \theta _2)]\nonumber\\
&=&|\cos(\theta_1-\theta_2)|^{\frac{1}{2}} \int d\alpha 'd\beta '\Delta \left (\beta,\alpha ';\theta _2\right )
\Delta \left (\alpha,-\beta ';\theta _1\right ){\widetilde W}(\alpha^\prime , \beta ^\prime|\rho)
\end{eqnarray}
In the special case $\theta_1=\theta_2=0$ this gives the Weyl function
\begin{eqnarray}
A(\alpha, \beta;0,0|\rho)={\widetilde W}(\beta ,-\alpha|\rho).
\end{eqnarray}
In the special case $\theta_1=\theta_2=\frac{\pi}{2}$ it gives the Wigner function
\begin{eqnarray}
A\left (\alpha, \beta;\frac{\pi}{2}, \frac{\pi}{2}|\rho\right )=W(\alpha, \beta|\rho).
\end{eqnarray}
Wigner functions quantify the noise, and Weyl functions quantify the correlations in a quantum system.
The Weyl function integrates a wavefunction with its displacement in phase space, and in this sense it describes correlations.
The widths of the Wigner function describe noise (both quantum and classical) in both the position and momentum.
The $\alpha, \beta$ in the Wigner function are position and momentum, while the $\alpha, \beta$ in the Weyl function are position and momentum increments,
related to correlations.

The quantity $A(\alpha, \beta; \theta _1, \theta_2|\rho)$ interpolates between the two, and shows that
correlations and uncertainties are different aspects of the same concept, which could be called `correlation-noise duality'.  
If $\theta _1, \theta_2$ are close to zero, this more general concept is close to correlations (because $A(\alpha, \beta; \theta_1, \theta_2|\rho)$
is close to the Weyl function), and if
$\theta _1, \theta_2$ are close to $\pi/2$, it is close to uncertainties (because $A(\alpha, \beta; \theta_1, \theta_2|\rho)$
is close to the Wigner function).
For general values of $\theta_1, \theta_2$ the $A(\alpha, \beta; \theta_1, \theta_2|\rho)$ interpolates between them, and quantifies
the noise-correlations duality.

\subsection{Marginal properties for $|A(\alpha, \beta; \theta_1, \theta_2|\rho)|^2$}\label{600}
In section \ref{p22} we gave the marginal properties for $U(\alpha, \beta; \theta _1, \theta _2)$.
Taking the trace of both sides of these equations with a density matrix $\rho$, we derive corresponding marginal properties for $A(\alpha, \beta; \theta_1, \theta_2|\rho)$.
Below we give marginal properties for $|A(\alpha, \beta; \theta_1, \theta_2|\rho)|^2$.

\begin{proposition}\label{P100}
\mbox{}
\begin{itemize}
\item[(1)]
Integration of $|A(\alpha, \beta; \theta_1, \theta_2|\rho)|^2$ with respect to $\alpha$ gives
\begin{eqnarray}\label{12}
 \int |A(\alpha, \beta; \theta_1, \theta_2|\rho)|^2 d\alpha &=&\sqrt{2} \pi |\cos(\theta_1-\theta_2)|
\nonumber \\&\times&\int dx \left \lvert \int d\alpha ' \Bra{x-\frac{\alpha '}{\sqrt{2}}} \rho \Ket{x+\frac{\alpha '}{\sqrt{2}}}
\Delta \left (\beta,\alpha ';\theta _2\right )\right \rvert ^2
\end{eqnarray}
\item[(2)]
Integration of $|A(\alpha, \beta; \theta_1, \theta_2|\rho)|^2$ with respect to $\beta$ gives
\begin{eqnarray}\label{12}
\int |A(\alpha, \beta; \theta_1, \theta_2|\rho)|^2 d\beta &=&\sqrt{2} \pi |\cos(\theta_1-\theta_2)|\nonumber\\&\times&
\int dp \left \lvert\int  d\beta ' \Bra{p-\frac{\beta '}{\sqrt{2}}} \rho \Ket{p+\frac{\beta '}{\sqrt{2}}}
\Delta \left (\alpha,-\beta ';\theta _1\right )\right \rvert ^2
\end{eqnarray}
\item[(3)]
Integration of $|A(\alpha, \beta; \theta_1, \theta_2|\rho)|^2$ with respect to both $\alpha$ and $\beta$ gives
\begin{eqnarray}\label{12}
 \int |A(\alpha, \beta; \theta_1, \theta_2|\rho)|^2 d\alpha d\beta &=&\pi |\cos(\theta_1-\theta_2)|{\rm Tr}(\rho ^2)
\end{eqnarray}
\end{itemize}
\end{proposition}
\begin{proof}
\begin{itemize}
\item[(1)]
Using Eq.(\ref{10}) we get
\begin{eqnarray}
 \int |A(\alpha,\beta;\theta_1, \theta_2|\rho)|^2 d\alpha &= &|\cos(\theta_1 - \theta_2)| \int d\alpha 'd\beta '\Delta \left (\beta,\alpha ';\theta_2\right ) \Delta \left (\alpha,-\beta ';\theta_1\right ) \widetilde{W}(\alpha ', \beta '|\rho) \nonumber \\
 &\times& \int d\alpha'' d\beta'' d\alpha \Delta \left (-\beta,\alpha'';-\theta_2\right ) \Delta \left (-\alpha,-\beta'';-\theta_1\right ) \widetilde{W}(\alpha'', \beta''|\rho)
\end{eqnarray}
Using Eqn.(6), integration with respect to $\alpha$ gives a delta function, and then integration with respect $\beta ''$ gives
\begin{eqnarray}
&&\int |A(\alpha,\beta;\theta_1, \theta_2|\rho)|^2 d\alpha \nonumber\\&=&|\cos(\theta_1 - \theta_2)| \int d\alpha' d\alpha'' d\beta'  \Delta \left (\beta,\alpha ';\theta_2\right ) \Delta \left (-\beta,\alpha'';-\theta_2\right )  \widetilde{W}(\alpha', \beta'|\rho)  \widetilde{W}(\alpha'',- \beta'|\rho)  \nonumber \\
&=& |\cos(\theta_1 - \theta_2)|\int d\alpha' d\alpha'' d\beta'  \Delta \left (\beta,\alpha ';\theta_2\right ) \Delta \left (-\beta,\alpha'';-\theta_2\right ) \nonumber \\
&\times &\int  dx\Bra{x-\frac{\alpha'}{\sqrt{2}}}\rho\Ket{x+\frac{\alpha'}{\sqrt{2}}} \exp(i\sqrt{2}\beta'x) \int  dy \Bra{y-\frac{\alpha''}{\sqrt{2}}}\rho\Ket{y+\frac{\alpha''}{\sqrt{2}}} \exp(-i\sqrt{2}\beta'y) \nonumber \\
\end{eqnarray}
Integration with respect to $\beta'$ gives a delta function, and changing variables, $x\rightarrow x - \frac{\alpha'}{\sqrt{2}}$ and $y\rightarrow y - \frac{\alpha''}{\sqrt{2}}$, we get
\begin{alignat}{1}
&\int |A(\alpha,\beta;\theta_1, \theta_2|\rho)|^2 d\alpha = \sqrt{2}\pi|\cos(\theta_1 - \theta_2)|\int d\alpha' dx \Bra{x-\frac{\alpha'}{\sqrt{2}}}\rho\Ket{x+\frac{\alpha'}{\sqrt{2}}} \Delta \left (\beta,\alpha ';\theta_2\right )\nonumber \\ 
&\times 
\int d\alpha''  \Bra{x-\frac{\alpha''}{\sqrt{2}}}\rho\Ket{x+\frac{\alpha''}{\sqrt{2}}}  \Delta \left (-\beta,\alpha'';-\theta_2\right )\nonumber \\
 &=\sqrt{2}\pi|\cos(\theta_1 - \theta_2)|\int dx \left | \int d\alpha' \Bra{x-\frac{\alpha'}{\sqrt{2}}}\rho\Ket{x+\frac{\alpha'}{\sqrt{2}}} \Delta \left (\beta,\alpha ';\theta_2\right )\right |^2 
\end{alignat}

\item[(2)]
The proof of this is similar to that above

\item[(3)]
From Eqs(\ref{1}),(\ref{10}) we get
\begin{alignat}{1}
 \int |A(\alpha, \beta; \theta_1, \theta_2|\rho)|^2 d\alpha d\beta 
 &=\lambda\int d\alpha d\beta  d\alpha 'd\beta ' d\alpha '' d\beta ''
\nonumber\\
&\times\exp\left [\frac {i}{2}(\alpha ''^2-\alpha '^2)\cot \theta_2 +\frac {i}{2}(\beta ''^2-\beta '^2)\cot \theta _1 \right ]\nonumber \\
&\times
\exp\left[\frac{i\beta (\alpha '+\alpha '')}{\sin \theta_2}-\frac{i\alpha (\beta '+\beta '')}{\sin \theta _1}\right]
{\rm Tr}[\rho D(\alpha ',\beta ')]{\rm Tr}[\rho D(\alpha '',\beta '')]
\end{alignat}
where
\begin{eqnarray}
\lambda &=&|\cos(\theta _1-\theta_2)|\left [\frac {1+i\cot{\theta _1}}{2\pi}\right ]^{1/2}\left [\frac {1-i\cot{\theta _1}}{2\pi}\right ]^{1/2}\left [\frac {1+i\cot{\theta_2}}{2\pi}\right ]^{1/2}
\left [\frac {1-i\cot{\theta_2}}{2\pi}\right ]^{1/2}\nonumber\\&=&\frac{|\cos(\theta _1-\theta_2)|}{4\pi ^2 \sin\theta _1\sin\theta_2}
\end{eqnarray}
Integration over $\alpha, \beta$ gives $\delta$-functions, and then we get Eq.(\ref{12}).
\end{itemize}
\end{proof}
A special case of Eq.(\ref{12}) for $\theta_1=\theta_2 =0$ and also $\theta_1=\theta_2 =\frac{\pi}{2}$, is
\begin{eqnarray}
\frac{1}{\pi}\int |W(\alpha, \beta)|^2d\alpha d\beta=\frac{1}{\pi}\int |{\widetilde W}(\alpha, \beta)|^2d\alpha d\beta={\rm Tr}(\rho ^2).
\end{eqnarray}

We next introduce the quantities
\begin{eqnarray}\label{12d}
\langle \langle \alpha  ^n \rangle \rangle=\frac{1}{\pi {\rm Tr}(\rho ^2)}\int \alpha ^n|A(\alpha, \beta; \theta_1, \theta_2|\rho)|^2 d\alpha d\beta;\;\;\;\;
\delta \alpha (\theta_1,\theta_2)=\left [\langle \langle \alpha  ^n \rangle \rangle-(\langle \langle \alpha   \rangle \rangle)^2\right ]^{1/2}
\end{eqnarray}
and similarly for $\delta \beta (\theta_1,\theta_2)$.
Such quantities have been introduced in \cite{V} for the special case of Wigner and Weyl functions.
It has been shown there that for pure states they are the usual uncertainties, but for mixed states they are different.
If $\theta_1, \theta_2$ are close to zero, $\delta \alpha (\theta_1,\theta_2)$, $\delta \beta (\theta_1,\theta_2)$ 
quantify correlations in position and momentum, and if $\theta_1, \theta_2$ are close to $\pi/2$, they quantify noise.

In ref\cite{V} it has been proved that $\delta \alpha (\frac{\pi}{2},\frac{\pi}{2})\delta \beta (0,0)\ge \frac{1}{2}$.
In the case of arbitrary angles considered here, we have not proved a similar inequality, but we study this product through an example. 
As an example, we plot the $\delta \alpha (\frac{\pi}{2},\theta_2)\delta \beta (0,0)$  as a function of $\theta_2$
in fig.\ref{f2}, for the quantum state described with the density matrix
\begin{eqnarray}\label{bbb}
\rho=\frac{1}{2}[\ket {\alpha _0,\beta _0}\bra {\alpha _0,\beta _0}+\ket {-\alpha _0,-\beta _0}\bra {-\alpha _0,-\beta _0}];\;\;\;\;\;
\ket {\alpha _0,\beta _0}=D(\alpha, \beta)\ket{0};\;\;\;\;\alpha _0=2;\;\;\;\;\;\beta _0=0.
\end{eqnarray}
We also plot the $\delta \alpha (\frac{\pi}{4},\frac{\pi}{4})\delta \beta (0,0)$  as a function of $p$
in fig.\ref{f3}, for the quantum state described with the density matrix
\begin{eqnarray}\label{bbb}
\rho=p\ket {\alpha _0,\beta _0}\bra {\alpha _0,\beta _0}+(1-p)\ket {-\alpha _0,-\beta _0}\bra {-\alpha _0,-\beta _0};\;\;\;\;\;
0\le p\le 1;\;\;\;\;\alpha _0=2;\;\;\;\;\;\beta _0=0.
\end{eqnarray}

\section{Moyal star formalism for bifractional Wigner functions}\label{AA3}
In this section we present the basic steps of the Moyal formalism for bifractional Wigner functions.
We start with a lema which is needed in proofs later.

\begin{lemma}\label{www}
For arbitrary states $\ket{\gamma}, \ket{\zeta}, \ket{\epsilon}, \ket{\delta}$
\begin{eqnarray}
\frac{1}{\pi \cos(\theta_1- \theta_2)} \int d\alpha d\beta \braket{\gamma|U^{\dagger}(\alpha,\beta;\theta_1,\theta_2)|\delta} \braket{\epsilon|U(\alpha,\beta;\theta_1,\theta_2)|\zeta} =\braket{\gamma|\zeta} \braket{\epsilon|\delta}
\end{eqnarray}
\end{lemma}
\begin{proof}
\begin{eqnarray}
&&\frac{1}{\pi \cos(\theta_1- \theta_2)} \int d\alpha d\beta \braket{\gamma|U^{\dagger}(\alpha,\beta;\theta_1,\theta_2)|\delta} \braket{\epsilon|U(\alpha,\beta;\theta_1,\theta_2)|\zeta}  \nonumber \\ &=&  \frac{1}{\pi}\int d\alpha d\beta d\alpha' d\beta' d\alpha'' d\beta''  \bra{\gamma} D(\alpha',\beta')\ket{\delta} \bra{\epsilon} D(\alpha'',\beta'')\ket{\zeta} 
\Delta(-\beta,\alpha';-\theta_2) \Delta(-\alpha,-\beta';-\theta_1)\nonumber\\&\times& \Delta(\beta,\alpha'';\theta_2) \Delta(\alpha,-\beta'';\theta_1)
=\frac{1}{\pi}\int d\alpha'' d\beta'' \bra{\gamma} D^{\dagger}(\alpha'',\beta'')\ket{\delta} \bra{\epsilon} D(\alpha'',\beta'')\ket{\zeta}
\end{eqnarray}
It is known \cite{M1,M2} that this is equal to $\braket{\gamma|\zeta}\braket{\epsilon|\delta}$.
\end{proof}

In the following proposition we express an operator $\Theta$ in terms of the $ U(\alpha,\beta;\theta_1,\theta_2)$, and the trace of a product of two operators
$\Theta _1 \Theta_2$, in terms of the corresponding bifractional Wigner functions.  
\begin{proposition}\label{p12}
\mbox{}
\begin{itemize}
\item[(1)]
\begin{eqnarray}\label{qqq}
\Theta =\frac{1}{\pi \cos(\theta_1- \theta_2)}\int d\alpha d\beta \;\;A^{\dagger}(\alpha,\beta;\theta_1,\theta _2|\Theta) U(\alpha,\beta;\theta_1,\theta_2) \nonumber \\
\end{eqnarray}

\item[(2)]
\begin{eqnarray}
{\rm Tr} (\Theta_1\Theta_2) = \frac{1}{\pi \cos(\theta_1-\theta_2)}\int d\alpha d\beta A(\alpha,\beta;\theta_1,\theta_2|\Theta_1) A^{\dagger}(\alpha,\beta;\theta_1,\theta_2|\Theta_2)
\end{eqnarray}
\end{itemize}
\end{proposition}
\begin{proof}
\mbox{}
\begin{itemize}
\item[(1)]
\begin{eqnarray}
&&\frac{1}{\pi \cos(\theta_1- \theta_2)}\int d\alpha d\beta \;\;A^{\dagger}(\alpha,\beta;\theta_1,\theta_2|\Theta) U(\alpha,\beta;\theta_1,\theta_2) \nonumber \\
&&=\frac{1}{\pi \cos(\theta_1- \theta_2)}\int d\alpha d\beta \;\;{\rm Tr}[U^{\dagger}(\alpha,\beta;\theta_1,\theta_2|\Theta)] U(\alpha,\beta;\theta_1,\theta_2)
\nonumber\\
 &&= \frac{1}{\pi}\int d\alpha d\beta d\alpha' d\beta' d\alpha'' d\beta'' \;\;{\rm Tr}[D^{\dagger}(\alpha',\beta')\Theta ] D(\alpha'',\beta'') 
\Delta(-\beta,\alpha';-\theta_2) \Delta(-\alpha,-\beta';-\theta_1)\nonumber\\&\times& \Delta(\beta,\alpha'';\theta_2) \Delta(\alpha,-\beta'';\theta_1) =
\frac{1}{\pi}\int d\alpha d\beta \;\;{\rm Tr}[D^{\dagger}(\alpha,\beta)\Theta] D(\alpha,\beta)=\Theta 
\end{eqnarray}
The last equality is a known relation (e.g., \cite{V1}).

A second proof based on lemma \ref{www}, is to consider the matrix elements of both sides with ordinary coherent states $\ket{z}$, $\ket{w}$:
\begin{eqnarray}\label{qw1}
\bra{z}\Theta \ket{w}=\frac{1}{\pi \cos(\theta_1- \theta_2)}\int d\alpha d\beta \;\;A^{\dagger}(\alpha,\beta;\theta_1,\theta _2|\Theta) \bra{z}U(\alpha,\beta;\theta_1,\theta_2)\ket{w} \nonumber \\
\end{eqnarray}
Also
\begin{eqnarray}\label{qw2}
A^{\dagger}(\alpha,\beta;\theta_1,\theta _2|\Theta) ={\rm Tr}[\Theta U^{\dagger}(\alpha,\beta;\theta_1,\theta_2)]=
\int \frac{d^2z}{\pi} \frac{d^2u}{\pi}\bra{\zeta}\theta \ket{u}\bra{u}U^{\dagger}(\alpha,\beta;\theta_1,\theta_2)\ket{\zeta}
\end{eqnarray}
Combining Eqs.(\ref{qw1}),(\ref{qw2}) and using lemma \ref{www}, which is valid for arbitrary states and therefore for coherent states,
we prove Eq.(\ref{qqq}).
\item[(2)]
Using Eq.(\ref{10}) we get
\begin{eqnarray}
&& \frac{1}{\pi \cos(\theta_1-\theta_2)}\int d\alpha d\beta A(\alpha,\beta;\theta_1,\theta_2|\Theta_1) A^{\dagger}(\alpha,\beta;\theta_1,\theta_2|\Theta_2) \nonumber \\
&=&  \frac{1}{\pi \cos(\theta_1-\theta_2)}  \int d\alpha' d\beta' d\alpha'' d\beta''   \widetilde{W}(\alpha',\beta'|\Theta_1) \widetilde{W}(\alpha'',\beta''|\Theta_2) \nonumber \\
&\times&\Delta(-\beta,\alpha';-\theta_2) \Delta(-\alpha,-\beta';-\theta_1) \Delta(\beta,\alpha'';\theta_2) \Delta(\alpha,-\beta'';\theta_1) 
\end{eqnarray}    
Using Eq.(\ref{11}) we get delta functions, which give
\begin{eqnarray}
&=& \frac{1}{\pi}\int d\alpha'' d\beta'' \; \widetilde{W}(-\alpha'',-\beta''|\Theta_1) \widetilde{W}(\alpha'',\beta''|\Theta_2) \nonumber \\
&=& \frac{1}{\pi}\int d\alpha'' d\beta'' \; \Braket{x+\frac{\alpha''}{\sqrt{2}}|\Theta_1|x-\frac{\alpha''}{\sqrt{2}}} \Braket{y-\frac{\alpha''}{\sqrt{2}}|\Theta_2|x+\frac{\alpha''}{\sqrt{2}}}  e^{i\sqrt{2}\beta''(y-x)}
\end{eqnarray}
Integrating with respect to $\beta''$, and changing variables, we show that:
\begin{eqnarray}
&=&\sqrt{2} \int d\alpha'' dx \; \Braket{x+\frac{\alpha''}{\sqrt{2}}|\Theta_1|x-\frac{\alpha''}{\sqrt{2}}} \Braket{x-\frac{\alpha''}{\sqrt{2}}|\Theta_1|x+\frac{\alpha''}{\sqrt{2}}}  \nonumber \\
&=&\int dk \braket{k|\Theta_1\Theta_2|k} = {\rm Tr}[\Theta_1\Theta_2].
\end{eqnarray}
\end{itemize}
\end{proof}

Given the $A(\alpha, \beta; \theta_1, \theta_2|\Theta _1)$ and $A(\alpha, \beta; \theta_1, \theta_2|\Theta _2)$ of two operators
$\Theta _1, \Theta _2$, the following proposition gives the $A(\alpha, \beta; \theta_1, \theta_2|\Theta _1\Theta _2)$ of their product.
\begin{proposition}\label{P15}
\begin{eqnarray}
&&A(\epsilon,\zeta;\theta_1,\theta_2|\Theta_1\Theta_2) \nonumber
\\&&=  \frac{1}{\pi[\cos(\theta_1-\theta_2)]^{1/2}} \int d\alpha'd\beta' d\alpha\; d\beta\; d\gamma\; d\lambda \;d\gamma' \;d\lambda' \; A^{\dagger}(\alpha,\beta;\theta_1,\theta_2|\Theta_1)\; A^{\dagger}(\alpha',\beta';\theta_1,\theta_2|\Theta_2) \nonumber \\ &\times& \Delta(\beta,\gamma;\theta_2) \Delta(\alpha,-\lambda;\theta_1) \Delta(\beta',\gamma';\theta_2) \Delta(\alpha',-\lambda';\theta_1) \Delta(\zeta,-(\gamma+\gamma');\theta_2)  \Delta(\epsilon,\lambda+\lambda';\theta_1) 
\nonumber \\ &&\times\exp[i\lambda\gamma'-i\gamma\lambda']
\end{eqnarray}
\end{proposition}
\begin{proof}
Eq.(\ref{10}) gives
\begin{eqnarray}\label{7b}
A(\epsilon,\zeta;\theta_1,\theta_2|\Theta_1\Theta_2)= |\cos(\theta_1-\theta_2)|^{1/2}\int d\epsilon'\;d\zeta' \; \Delta(\zeta,\epsilon';\theta_2) \Delta(\epsilon,-\zeta';\theta_1) \;{\rm Tr}[D(\epsilon',\zeta')\Theta_1\Theta_2]
\end{eqnarray}
Using Eq.(\ref{qqq}) we get
\begin{eqnarray}
\Theta _1\Theta _2
&=& \frac{1}{[\pi\cos(\theta_1-\theta_2)]^2}\int d\alpha'd\beta' d\alpha d\beta \;A^{\dagger}(\alpha,\beta;\theta_1,\theta_2|\Theta_1)\; A^{\dagger}(\alpha',\beta';\theta_1,\theta_2|\Theta_2)\nonumber\\&\times&  \;U(\alpha,\beta;\theta_1,\theta_1) \;U(\alpha',\beta';\theta_1,\theta_2)
\end{eqnarray}
Therefore
\begin{eqnarray}\label{7b}
A(\epsilon,\zeta;\theta_1,\theta_2|\Theta_1\Theta_2) &=& 
\frac{[{\cos(\theta_1-\theta_2)}]^{-3/2}}{ \pi^2} \int d\alpha'd\beta' d\alpha d\beta \; d\epsilon' d \zeta' A^{\dagger}(\alpha,\beta;\theta_\alpha,\theta_\beta|\Theta_1)\; A^{\dagger}(\alpha',\beta';\theta_1,\theta_2|\Theta_2) \nonumber \\ &\times& {\rm Tr}[U(\alpha,\beta;\theta_1,\theta_2) U(\alpha',\beta';\theta_1,\theta_2)D(\epsilon',\zeta')] \Delta(\zeta,\epsilon';\theta_2) \Delta(\epsilon,-\zeta';\theta_1)\nonumber\\&=&
\frac{[{\cos(\theta_1-\theta_2)}]^{-3/2}}{ \pi^2} \int d\alpha'd\beta' d\alpha d\beta \; d\epsilon' d \zeta'  d\gamma\; d\lambda \;d\gamma' \;d\lambda' \;A^{\dagger}(\alpha,\beta;\theta_1,\theta_2|\Theta_1)\;\nonumber \\ &\times&  A^{\dagger}(\alpha',\beta';\theta_1,\theta_2|\Theta_2) 
\Delta(\beta,\gamma;\theta_2) \Delta(\alpha,-\lambda;\theta_1) \Delta(\beta',\gamma';\theta_2) \Delta(\alpha',-\lambda';\theta_1)\nonumber\\&\times& 
{\rm Tr}[D(\gamma, \lambda )D(\gamma ', \lambda ')D(\epsilon',\zeta')] \Delta(\zeta,\epsilon';\theta_2) \Delta(\epsilon,-\zeta';\theta_1)
\end{eqnarray}
But
\begin{eqnarray}
{\rm Tr} [D(\gamma,\lambda)D(\gamma',\lambda')D(\epsilon',\zeta')] = \pi \delta(\epsilon'+\gamma+\gamma')\delta(\zeta'+\lambda+\lambda')\exp[i(\lambda \gamma '-\gamma \lambda ')].
\end{eqnarray}
Inserting this in Eq.(\ref{7b}) we prove the proposition.
\end{proof}

\section{Generalized Berezin formalism}\label{AA4}

The Berezin formalism\cite{BR1,BR2,BR3,BR4} represents an operator $\Theta$ with the analytic function $L(z,w^*;\theta_1,\theta_2|\Theta)$ defined below.
It shows that the $L(z,z^*;\theta_1,\theta_2|\Theta _1 \Theta _2)$ of the product of two operators $\Theta _1 \Theta _2$, can be expanded as a Taylor series,
where the first term is the product $ L(z,z^*;\theta_1,\theta_2|\Theta_1) L(z,z^*;\theta_1,\theta_2|\Theta_2) $ (which is classical in the sense that it is commutative),
and the other terms are quantum corrections (and go to zero in the limit $\hbar \rightarrow 0$).
The Laplacian used in the standard Berezin formalism, is replaced here with the `bifractional Laplacian\rq{} defined below.

\begin{lemma}
For $K>0$
\begin{eqnarray}\label{8x}
&&\frac{1}{2\pi \cos(\theta_1 - \theta_2)} \int d\alpha' d\beta' F(\alpha',\beta') K \exp\left\{ -\frac{K[d(\alpha -\alpha ',\beta -\beta '|\theta_1,\theta_2)]^2 } {\cos^2(\theta_1-\theta_2)}\right \} 
\nonumber\\&&= \frac{1}{2} \left [\exp{\left(  \frac{\Delta_{(\alpha,\beta|\theta_1,\theta_2)}}{4K} \right) } F(\alpha,\beta) \right ], 
\end{eqnarray}
where
\begin{eqnarray}\label{c45}
\Delta_{(\alpha,\beta|\theta_1,\theta_2)} =  \frac{\partial^2}{\partial^2 \alpha} + \frac{\partial^2}{\partial^2 \beta} - 2\frac{\partial^2}{\partial \alpha \partial \beta}\sin(\theta_1 - \theta_2).
\end{eqnarray}
If we replace $\alpha\rq{}, \beta \rq{}$ with $w,w^*$ given in Eq.(\ref{rrr})
(and $\alpha, \beta $ with $z,z^*$), then Eq.(\ref{c45}) can be re-written as
 \begin{eqnarray}\label{4cv}
 &&\frac{1}{2\pi} \int d^2 w F(w,w^*) K  \exp{ \left[-K{|w(\theta_1,\theta_2) - z(\theta_1,\theta_2)|}^2 \right]} = \frac{1}{2} \left [\exp{\left( \frac{\Delta_{(z,z^*|\theta_1,\theta_2)}}{4K} \right)} F(z,z^*) \right ],
 \end{eqnarray}
 where $\Delta_{(z,z^*|\theta_1,\theta_2)}$ is the `bifractional Laplacian\rq{}
  \begin{eqnarray}
 &&\Delta_{(z,z^*|\theta_1,\theta_2)} = 4\frac{\partial ^2}{\partial z \partial z^*} - 2i\left[\frac{\partial^2}{\partial^2 z} - \frac{\partial^2}{\partial^2 z^*} \right] \sin(\theta_1-\theta_2).
\end{eqnarray}
\end{lemma}
\begin{proof}
The proof of Eq.(\ref{8x}) is based on a Fourier transform of both sides (it is lengthy but straightforward). 

\end{proof}
Let
\begin{eqnarray}
L(z,w^*;\theta_1,\theta_2|\Theta) = \exp \left [\frac{1}{2}|z|^2+\frac{1}{2}|w|^2- zw^*\right ]  
\braket{z^*(\theta _1, \theta _2)|\Theta|w^*(\theta _1, \theta _2)}
\end{eqnarray}
This is an analytic function of $w^*(\theta _1, \theta _2)$ and $z(\theta _1, \theta _2)$.

\begin{proposition}\label{p13}
\begin{eqnarray}
L(z,z^*;\theta_1,\theta_2|\Theta_1\Theta_2) =   \frac{1}{2}  \left [\exp{\left( \frac{\Delta_{(\zeta,\zeta^*|\theta_1,\theta_2)}}{4} \right)} L(z,\zeta^*;\theta_1,\theta_2|\Theta_1) L(\zeta,z^*;\theta_1,\theta_2|\Theta_2)\right]_{\zeta = z}
\end{eqnarray}
Taylor expansion gives
\begin{eqnarray}
L(z,z^*;\theta_1,\theta_2|\Theta_1\Theta_2) &=&  L(z,z^*;\theta_1,\theta_2|\Theta_1) L(z,z^*;\theta_1,\theta_2|\Theta_2) 
\nonumber\\&&+   \frac{\partial L(z,z^*;\theta_1,\theta_2|\Theta_1)}{2 \partial z^*} \frac{\partial L(z,z^*;\theta_1,\theta_2|\Theta_1)}{\partial z} \nonumber \\
&& +\left [i\sin(\theta_1-\theta_2)\frac{\partial^2 L(z,z^*;\theta_1,\theta_2|\Theta_1) L(z,z^*;\theta_1,\theta_2|\Theta_2)}{4 \partial^2 z^*}  \right] \nonumber \\
&& -\left[i\sin(\theta_1-\theta_2)\frac{\partial^2 L(z,z^*;\theta_1,\theta_2|\Theta_1) L(z,z^*;\theta_1,\theta_2|\Theta_2)}{4 \partial^2 z} \right]
+...
\end{eqnarray}
\end{proposition}
\begin{proof}
For two arbitrary operators $\Theta_1, \Theta_2$ we have that,
\begin{align}
&L(z,z^*;\theta_1,\theta_2|\Theta_1\Theta_2) = \int d^2 w D_E(z,w;\theta_1,\theta_2) L(z,w^*;\theta_1,\theta_2|\Theta_1) L(w,z^*;\theta_1,\theta_2|\Theta_2) \nonumber \\
&D_E(z,w;\theta_1,\theta_2) = |\langle {z(\theta _1, \theta _2)}\ket{w(\theta _1, \theta _2)}|^2
\end{align}
 Then using Eq.(\ref{4cv}), we prove Proposition \ref{p13}.
\end{proof}

\section{Discussion}
We have studied bifractional transforms, and their application in the area of phase space methods.
They provide a two-parameter ($\theta _1,\theta _2$) interpolation between other known quantities.
We have explained that they do not form a group and we used groupoids to describe their mathematical structure.

The work generalizes the traditional concept of phase space. The Wigner function $W(\alpha, \beta)$ describes the quantum noise in the position of a particle in the phase space $\alpha-\beta$.
The Weyl function describes quantum correlations in the space $\alpha^\prime -\beta ^\prime$
of position and momentum increments.
$(\alpha, \beta ^\prime)$ are dual variables in the Fourier transform sense, and the same is true for
$(\alpha ^\prime , \beta)$. Through fractional Fourier transforms, we work in an intermediate phase space
$\alpha _{\theta _1}-\beta _{\theta _2}$, where $\alpha _{\theta _1}$ is in the plane $\alpha-\beta ^\prime$, and $\beta _{\theta _2}$ is in the plane $\alpha^\prime\boldsymbol {-}\beta$.
When $\theta _1=\theta _2=0$, the $\alpha _{\theta _1}- \beta _{\theta _2}$ is the
$\alpha-\beta$ position-momentum phase space, associated with the Wigner function and quantum noise.
 When $\theta _1=\theta _2=\frac{\pi}{2}$, the $\alpha _{\theta _1}- \beta _{\theta _2}$ is the
 dual phase space
$\alpha^\prime -\beta^\prime $ of position increment and momentum increment, associated with the Weyl function and quantum correlations.
Our intermediate phase space is related to novel intermediate quantities between quantum correlations and quantum noise, and reveals deep links between them.

Using bifractional transforms we have defined bifractional coherent states.
Their analyticity properties and their resolution of the identity have been presented in proposition \ref{P1}.
They are the counterparts in the phase space $\alpha _{\theta _1}-\beta _{\theta _2}$,
of the standard coherent states in the phase space $\alpha-\beta$.

We have also defined  bifractional Wigner functions $A(\alpha, \beta;\theta_1,\theta_2|\rho)$.
We have studied their properties, and interpreted them physically as quantities which interpolate between quantum noise and quantum correlations.
We have also studied the Moyal star formalism for bifractional Wigner functions, and the corresponding  Berezin formalism (proposition \ref{p13}). 
This 
 provides a complete study of the $\alpha _{\theta _1}-\beta _{\theta _2}$  phase space
 that we introduced in this paper.

\newpage

\begin{figure}
\caption{A non-orthogonal frame of $(\theta _1,\theta _2)$ axes in phase space}
\includegraphics[width=0.5\textwidth]{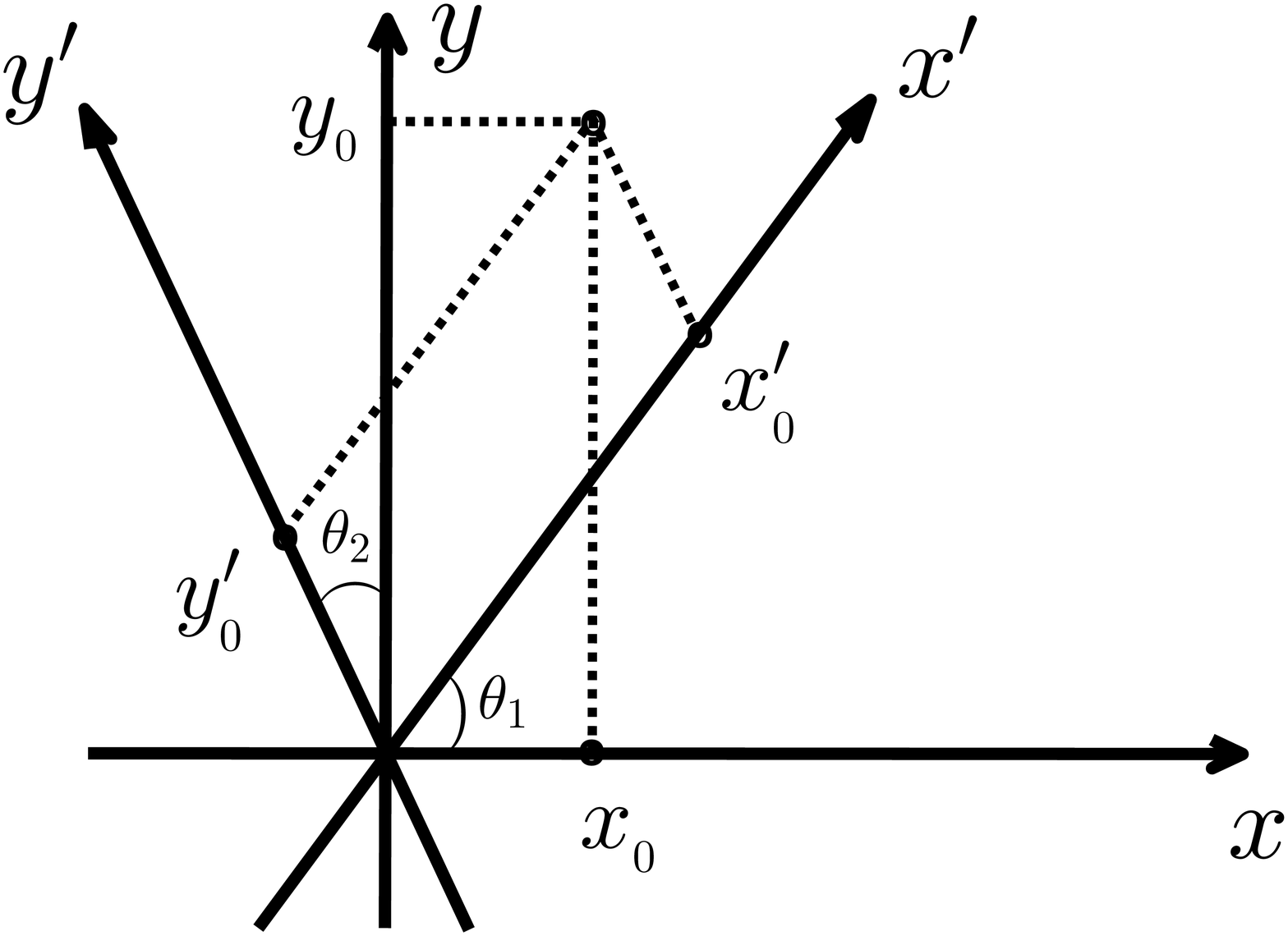}
\label{f1}
\end{figure}

\begin{figure}
\caption{The $\delta \alpha (\frac{\pi}{2},\theta_2)\delta \beta (0,0)$ as a function of $\theta_2$
(in rads), for the density matrix of Eq.(\ref{bbb}).}
\includegraphics[width=0.5\textwidth]{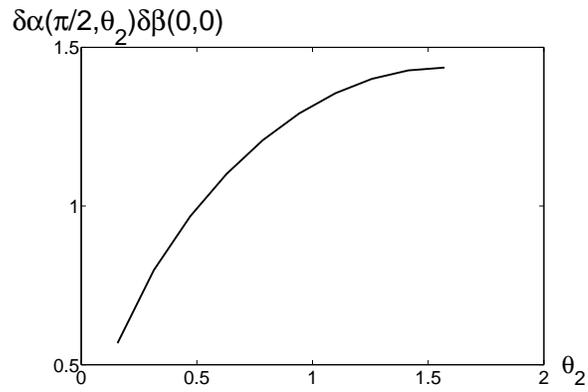}
\label{f2}
\end{figure}

\begin{figure}
\caption{The $\delta \alpha (\frac{\pi}{4},\frac{\pi}{4})\delta \beta (0,0)$ as a function of $p$, for the density matrix of Eq.(\ref{bbb}).}
\includegraphics[width=0.5\textwidth]{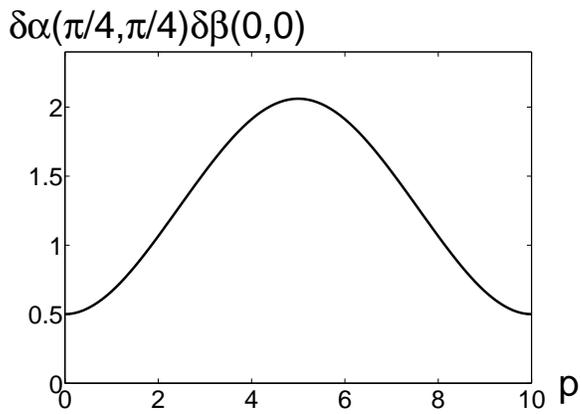}
\label{f3}
\end{figure}

\end{document}